\documentclass[journal,10pt]{IEEEtran}

\usepackage{amssymb}
\usepackage{amsmath}
\usepackage{amsfonts}
\usepackage{graphicx}
\usepackage{subfigure}
\usepackage{cite}
\usepackage{enumerate}
\usepackage{url}
\usepackage{bm}
\usepackage{stfloats}
\usepackage{tabularx}
\usepackage[ruled]{algorithm2e}

\begin{document}

\newtheorem{definition}{\bf Definition}
\newtheorem{theorem}{\bf Theorem}
\newtheorem{lemma}{\bf Lemma}
\newtheorem{property}{\bf Property}
\newtheorem{corollary}{\bf Corollary}
\newtheorem{remark}{\bf Remark}
\newtheorem{step}{\bf Step}

\title{Resource Allocation in Wireless Powered Relay
Networks: A Bargaining Game Approach}

\IEEEoverridecommandlockouts

\author{\IEEEauthorblockN{Zijie Zheng\IEEEauthorrefmark{1},
Lingyang Song\IEEEauthorrefmark{1},
Dusit Niyato\IEEEauthorrefmark{2},
and Zhu Han\IEEEauthorrefmark{3},}\\
\normalsize{\IEEEauthorblockA{\IEEEauthorrefmark{1}School of Electrical Engineering and Computer Science, Peking University, Beijing, China\\}}
\normalsize{\IEEEauthorblockA{\IEEEauthorrefmark{2}School of Computer Engineering,
Nanyang Technological University, Singapore\\}}
\IEEEauthorblockA{\IEEEauthorrefmark{3}Electrical and Computer Engineering Department, University of Houston, Houston, TX, USA\\}
%\IEEEauthorblockA{Email: \IEEEauthorrefmark{1}\{zijie.zheng, lingyang.song\}@pku.edu.cn, \IEEEauthorrefmark{2}\{lgao, jwhuang\}@ie.cuhk.edu.hk}
%\thanks{This work is supported in part by the National 973 Project under Grant 2013CB336700, by the National Natural Science Foundation of China under Grants 61222104 and U1301255.}
%\thanks{This work is also supported by the General Research Funds (Project Number CUHK 412713 and 14202814) established under the University Grant Committee of the Hong Kong Special Administrative Region, China.}
\thanks{Part of the material in this paper has been accepted by IEEE International Conference on Communications,
Malaysia, May 2016.~\cite{ZSNH-2016}}
}

\maketitle

\begin{abstract}
Simultaneously information and power transfer in mobile relay networks have recently emerged, where the relay can harvest the radio frequency (RF) energy and then use this energy for data forwarding and system operation. Most of the previous works do not consider that the relay may have its own objectives, such as using the harvested energy for its own transmission instead of maximizing transmission of the network. Therefore, in this paper, we propose a Nash bargaining approach to balance the information transmission efficiency of source-destination pairs and the harvested energy of the relay in a wireless powered relay network with multiple source-destination pairs and one relay. We analyze and prove that the Nash bargaining problem has several desirable properties such as the discreteness and quasi-concavity, when it is decomposed into three sub-problems: the energy transmission power optimization, the power control for data transmission and the time division between energy transmission and data transmission. Based on the theoretical analysis, we propose an alternating power control and time division algorithm to find a suboptimal solution. Simulation results clearly show and demonstrate the properties of the problem and the convergence of our algorithm.
\end{abstract}

%%%%%%%%%%%%%%%%%%%%%%%%%%%%%%%%%%%%%%%%%%%%%%%
\vspace{-1mm}
\section{Introduction}
\vspace{-1mm}
Harvesting energy from the natural resources such as solar energy, wind energy and thermal energy has been proved to be a promising method to prolong the lifetime for networks without fixed power supply, such as wireless sensor networks~\cite{ZH-2013}. However, uncertainty and uncontrollability of environment usually result in that energy harvesting has less stable energy supply and might sometimes insufficient causing the outage of the equipments in the networks~\cite{EH-2006}. To overcome uncertainty and randomness of energy harvesting from environment, the technique of radio frequency~(RF) energy transfer is adopted recently and applied widely in wireless powered networks, such as in wireless sensor networks~\cite{NKA-2010}, in wireless
body area networks~\cite{ZJZZWC-2010} and in wireless charging systems~\cite{LNWKH-2015}. With RF energy transfer, wireless receivers can harvest energy through converting received signals from wireless transmitters into electricity and store it in batteries~\cite{SWIPT-2008}. Since the information can also be transmitted using radio frequency in wireless networks, the concept of simultaneous wireless information and power transfer~(SWIPT) has become a promising approach for energy and information delivery in wireless networks~\cite{SWIPT-2008}~\cite{SURVEY-2014}. \par
More recently, besides information and power transfer from a wireless node to another, SWIPT has also been extended to wireless relay networks~\cite{MM-2010,NZDK-2013,NZDK-2014,ZZH-2013,LGKK-2015,ZPZL-2015,MSS-2015,ZZ-2015,LZQ-2014}. It is known that relays can expand the coverage of the wireless networks, where the relays help forward the sources' information to the destinations far away. However, in reality, some relays, e.g., wireless sensors and mobile phones, may have limited battery reserves and need external energy supply to maintain the circuits system activeness and help transmit the data~\cite{MM-2010}. With SWIPT, relays can at first harvest energy from sources or other wireless transmitters, and then forward information to destinations with the harvested energy~\cite{NZDK-2013}.\par
The difficulty in realizing SWIPT effectively in such wireless powered relay networks is that practical circuits on the relays cannot harvest energy and extract data
from wireless signals at the same time~\cite{ZZH-2013}\cite{NZDK-2014}. Therefore, one of the key challenges is balancing the tradeoff between the quality of information transmission and wireless energy transfer. For example, more time or frequency resource allocated to source-to-relay information transmission results in less time to harvest energy by the relays, which might lead to relays' turnoff or energy outage. To maximize the throughput and avoid the energy outage of the relays, the time switching-based relaying~(TSR) protocol and power splitting-based relaying~(PSR) protocol were proposed in~\cite{NZDK-2013} and~\cite{NZDK-2014}, respectively, deciding the energy harvesting time ratio and power splitting ratio in the one-source-one-relay network. To cope with the system that multiple sources, authors in~\cite{LGKK-2015} designed different encode and decode protocols for multiple source energy and data transmission and compared their signal-to-error-ratios~(SERs). Moreover, TSR and PSR protocols were extended in~\cite{ZPZL-2015} to deal with the antenna selection problem for energy harvesting and information transmission, when the relay was equipped with a MIMO system. The relay selection problem was studied in~\cite{MSS-2015} when the networks include multiple relays, where the authors maximized the capacity under an energy transfer constraint. The concept of self-looping energy recycling was proposed very recently with full-duplex relays. The capacity was derived with the TSR protocol in~\cite{ZZ-2015} with full-duplex relays and the average rate was optimized in~\cite{LZQ-2014} when in combination of full-duplex relays and the MIMO system. \par
As mentioned above, most of the existing works considered the information transmission and energy transfer tradeoff to achieve a common objective such as maximizing the total capacity of the system. This setting is suitable only when all nodes belong to the same authority. However, similar to the conventional wireless relay networks~\cite{WHL-2009,ZDGDC-2014,add1-2}, in reality, nodes in the wireless powered relay networks have their own selfish goals, which usually conflict with each other. As for wireless powered relay networks, each source-destination pair cares about its information transmission efficiency,~i.e.~bits successfully transmitted per Joule energy cost. As for the relay, when the energy can be accumulated in the battery~\cite{KI-2014,CXL-2014,DPEP-2014}, the relay cares mostly about its energy benefit, i.e., the harvested energy left after helping sources transmit information. The residual harvested energy can be used by the relay, e.g., for its own data transmission or system operation. Thus, source-destination pairs and the relay hold contradicted goals where sources try to guarantee their transmitted energy to the relay are used to transmit sources data as much as possible. However, the relay wants to uses less harvested energy to help sources and to keep more for the relay's own uses. The separated and contradicted individual objectives of source-destination pairs and the relay require us to investigate how they negotiate and bargain with each other. Thus, it is natural to apply game theory to balance the objectives among different nodes~\cite{game-1994,game-2011,add2-1,add2-2}.\par
To balance the information transmission efficiency for source-destination pairs and the residual harvested energy for the relay, in this paper, we propose a Nash bargaining approach to obtain the Nash bargaining solution~\cite{bargaining-1999}, through optimizing power control and time allocation. We prove the non-convexity of the bargaining problem and simplify it through decomposing the problem into three sub-problems: the energy transmission power optimization, the power control for information transmission and the time division between information transmission and energy transmission. Based on the theoretical analysis of the each problem, we design an alternating power control and time division algorithm to reach a suboptimal solution of the Nash bargaining solution.\par
The main contributions of this paper are as follows:
\begin{itemize}
\item We are the first to introduce a Nash bargaining approach in wireless powered relay networks to balance the tradeoff between the energy efficiency of information transmission of source-destination pairs and the residual harvested energy of the relay.
\item The bargaining problem decomposition is not only physically meaningful but also mathematically tractable, where the subproblem of energy harvesting power optimization is proved to be discretizable. Moreover, the subproblems of time division and information transmission are all proved to be quasi-concave.
\item Simulation results corroborate the properties or subproblems and the convergence of our algorithm. In addition, the utility imbalance in the Nash bargaining solution is also illustrated through simulation.
\end{itemize}
\par
The rest of this paper is organized as follows. The system model and problem formulation are presented in Section II. In Section III, we theoretically analyze and decompose the problem. Based on the analysis, we present our algorithm in Section IV. Section V includes the numerical simulation results. Finally, we conclude this paper in Section VI.\par

\vspace{-2mm}
\section{System Model}
\begin{figure}[!t]
\centering
\includegraphics[width=3.3in]{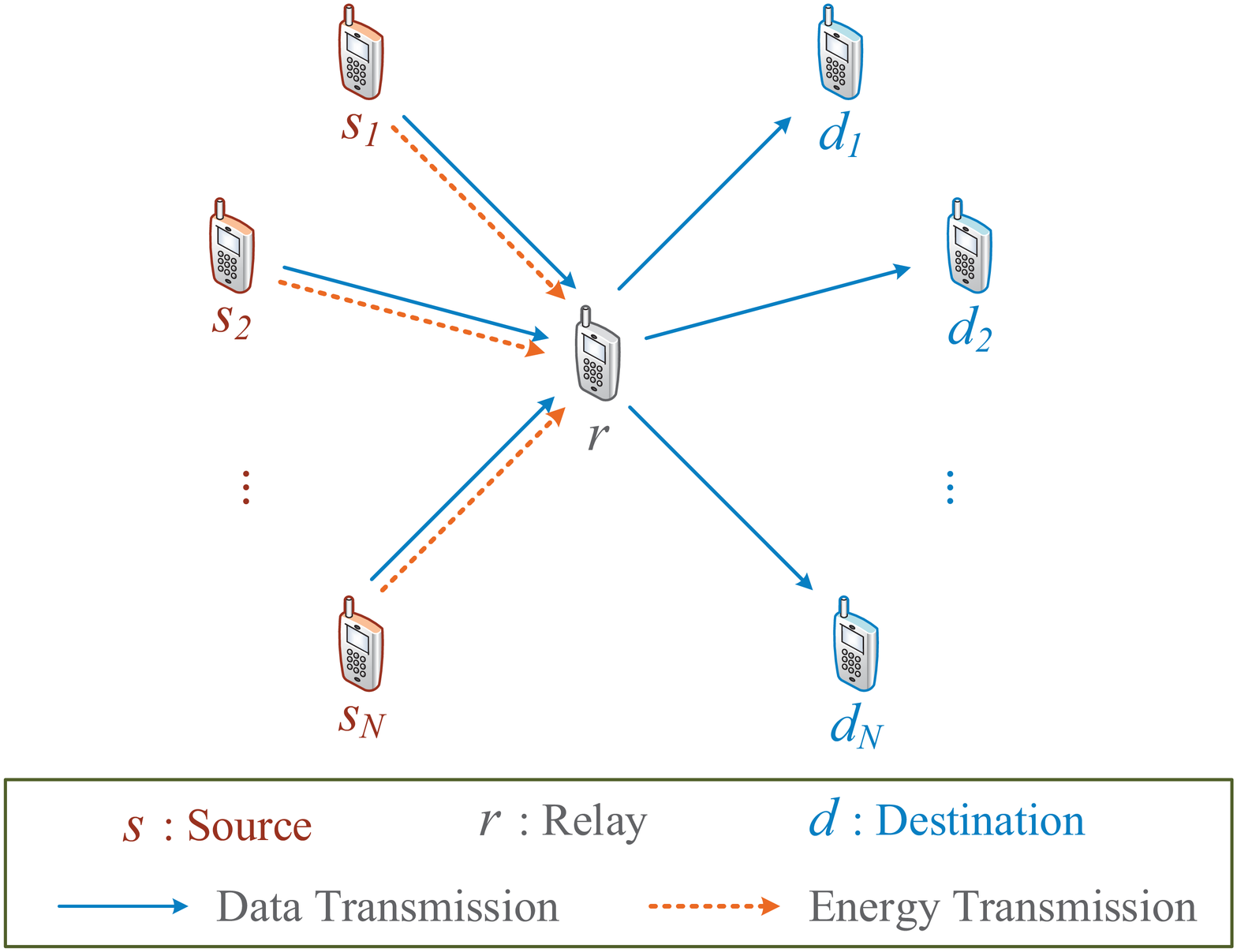}
\caption{System model for wireless powered relay networks.}
\label{fig:sys}
\end{figure}
We consider a wireless powered communication system as shown in Fig.~\ref{fig:sys}. For simplicity, we suppose that the system includes $N$ source-destination pairs denoted by $\mathcal{(S,D)}=\{(s_1,d_1),(s_2,d_2),\ldots,(s_N,d_N)\}$ and an RF-powered cooperative relay, denoted by $r$. More complicated problems in practical networks with multiple relays, such as relay selection, interference control and spectrum resource allocation are not considered in this paper, which is left for the future work. Sources transmit data to their destinations with the
help of this relay (which does not have its own data to transmit). In this paper, the relay is RF-powered~\cite{SURVEY-2014}. The half-duplex time-switching relaying protocol~(TSR)~\cite{NZDK-2013} is adopted to support the relay's energy harvesting from sources and data forwarding. Furthermore, the direct link is assumed to be ignored between each source-destination pair~\cite{CHXGCHLY-2013}. More details on the channel model, TSR protocol, utility of each node and problem formulation are given as follows.
\subsection{Channel Model}
The channel between each pair of nodes is modeled as a quasi-static flat fading channel~\cite{KTS-2012}, where the channel is constant over a block time $T$ and independent from one block to the next. Without loss of generality, the channel fading in each block follows a Rayleigh distribution. The use of such
channels is motivated by prior research in the wireless powered networks~\cite{ZZH-2013,NZDK-2014,HZ-2012}. For simplicity, it is also assumed that the channel state information is available at the relay and destinations.\par
\begin{figure*}[!t]
\centering
\includegraphics[width=5.5in]{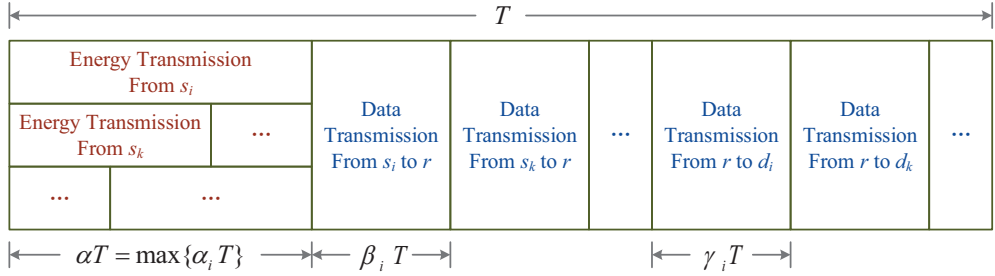}
\caption{TSR protocol.}
\label{fig:dftsr}
\end{figure*}
\subsection{TSR Protocol}
Fig.~\ref{fig:dftsr} describes the block structure in the TSR protocol for energy harvesting and information forwarding at the relay. It is a direct extension from the protocol in~\cite{NZDK-2013} and~\cite{NZDK-2014} in support of the case where the relay is able to assist multiple sources inside one time block $T$. The relay first harvests the energy from its assisting sources simultaneously. Then the relay receives signals from sources and then decodes and forwards~(DF) the data to its assisting destinations sequentially with the harvested energy~\cite{NZDK-2014,AE-2016,NNZKD-2016,FZJQS-2015}. As shown in Fig.~\ref{fig:dftsr}, $\alpha T$ denotes the fraction of the block time where the relay $r$ harvests energy, which is determined by the energy transmission time of each of its assisting source~$s_i$, denoted by $\alpha_i T$. We have $\alpha T=\max_{ (s_i, d_i) \in (\mathcal{S},\mathcal{D}) }\alpha_i T$, which indicates that the relay does not begin transmitting data before accomplishing energy harvesting from all its assisting sources. $\beta_iT$ is the time for source $s_i$ to relay $r$ data transmission, and $\gamma_iT$ is the time for information forwarding from $r$ to $d_i$. Thus, the following constraint is naturally satisfied:
\begin{equation}
\max_{\alpha_i}~(\alpha_iT)+\sum_{ (s_i, d_i) \in (\mathcal{S},\mathcal{D}) } (\beta_{i}+\gamma_{i})T=T.
\label{equ:timesplitting}
\end{equation}
In this paper, we assume that the following constraint must also be satisfied to guarantee the quality of service for data transmission of each
source-destination pair. Each source-destination pair has at least $\theta_0 T$ unit time for data transmission in this system:
\begin{equation}
(\beta_i+\gamma_i)T\geq\theta_0T.
\label{equ:timeconstraint}
\end{equation}
\subsection{Utility of Each Node}
\subsubsection{Source-Destination Pair Utility} For each source-destination pair $(s_i,d_i)$, we adopt the \emph{data transmission efficiency}~\cite{NLS-2013} as its gain, which is defined as the total average number of bits successfully transmitted
from $s_i$ to $d_i$ per Joule unit of energy consumed by source $s_i$, in each block time. Assuming that the relay
The data transmission efficiency $U_i^S$ for ($s_i,d_i$) can
be calculated as follows:
\begin{equation}
\small
\left\{
\begin{aligned}
&U_i^S=\frac{C_i\cdot T}{P^{s1}_i\alpha_iT+P^{s0}_i\beta_iT}=\frac{C_i}{P^{s1}_i\alpha_i+P^{s0}_i\beta_i},\\
&C_i= \min \left\{\beta_i\log\left(1+\frac{P^{s0}_i|h_{i}|^2}{\sigma^2}\right),\gamma_i\log\left(1+\frac{P^r|g_{i}|^2}{\sigma^2}\right)\right\},\\
\end{aligned}
\right.
\label{equ:sourceutility}
\end{equation}
where $C_i$ is the average capacity for source $s_i$ transmitting data to $d_i$ via relay $r$ with the TSR protocol, $P^{s0}_i$ is the transmission power of source $s_i$ for data transmission and $P^{s1}_i$ is the transmission power of source $s_i$ for energy transmission, which can be different from $P^{s0}_i$, $P^r$ is the transmission power of relay $r$, $h_{i}$ is the channel gain from $s_i$ to $r$, and $g_{i}$ denotes the channel gain from $r$ to $d_i$. Since either in the energy transmission, or in the data transmission, the transmission power cannot exceed the physical power upper bound of each equipment, we have the power constraints as follows:
\begin{equation}
\left\{
\begin{aligned}
&0\leq P^{s1}_i\leq P^{s}_0,\\
&0\leq P^{s0}_i\leq P^{s}_0,\\
&0\leq P^r\leq P^r_0,\\
\end{aligned}
\right.
\label{equ:powerconstraint}
\end{equation}
where $P^{s}_0$ is the maximum energy~(or information) transmission power for each source, and $P^{r}_0$ is the maximum transmission power for the relay.\par
\subsubsection{Relay Utility} For relay node $r$, its utility is the residual energy after energy harvesting and data forwarding, which can be used to execute other applications for itself, e.g., sensing or transmitting data in sensor networks.
The harvested energy at relay $r$ per block time can be calculated as
\begin{equation}
E=\eta \sum_{(s_i,d_i)\in (\mathcal{S},\mathcal{D})} \alpha_iT P^{s1}_i|h_{i}|^2,
\label{equ:eh}
\end{equation}
where $\eta\in(0,1)$ is the energy conversion efficiency which depends on the energy harvesting circuits~\cite{ZZH-2013}.
The energy cost in relay $r$ per each block time can be calculated as follows:
\begin{equation}
\varphi= \sum_{(s_i,d_i)\in (\mathcal{S},\mathcal{D})} \left( \gamma_iT P^r\right)+E_0,
\label{equ:relaycost}
\end{equation}
where the first term is the energy cost for data transmission and $E_0$ is a constant, which denotes the energy cost for signal decoding and recoding in the assistance of data forwarding. Hence, the utility of relay $r$ is denoted by $U^R$, which is calculated as follows:
\begin{equation}
U^R=E-\varphi.
\label{equ:relayutiliy}
\end{equation}

\subsection{Game Theoretic Problem Formulation}
Both the source-destination pairs and the relay aim to achieve the highest utility. However, the more available energy at the relay requires more time for energy transmission (with larger $\{\alpha_i\}$) and less time for data forwarding~( with smaller $\{\beta_i\}$ and $\{\gamma_i\}$). This usually results in the lower utilities for the source-destination pairs. Apart from the influence of time division, the transmission power from sources (such as $P^{s1}_i$ and $P^{s0}_i$ in (\ref{equ:sourceutility}) and (\ref{equ:eh})) and the transmission power from the relay (such as $P^{r}$ in (\ref{equ:sourceutility}) and (\ref{equ:eh})) jointly affect the utilities of source-destination pairs and the relay. Since the utilities of the sources and the utility of the relay are obviously contradicted, a Nash bargaining game~\cite{bargaining-1999} is adopted in this paper and the Nash bargaining solution~\cite{bargaining-1999} is considered as a reasonable solution to balance the utilities of sources and the utility of the relay.\par
In the Nash-bargaining game, the source-destination pairs and the relay are modeled as $N+1$ players, where for convenience the relay is denoted as the $({N+1})^{th}$ player. The strategy of each player (each source-destination pair and the relay) consists of all the variables that may influence the value of the player's utility. Even though the utility of each player does not directly contain all variables as in~(\ref{equ:sourceutility}) and~(\ref{equ:relayutiliy}), the constraint in~(\ref{equ:timesplitting}) and the utility function in~(\ref{equ:sourceutility}) make all variables correlated. Therefore, each player needs to consider all the variables in its strategy. That is to say, each player's strategy consists of the transmission power of sources $\{P^{s1}_i\}$, $\{P^{s0}_i\}$, the transmission of the relay $P^{r}$, and the time division variables $\{\alpha_i\}$, $\{\beta_i\}$, and $\{\gamma_i\}$. The strategy of the $n$'s player is denoted by $\Upsilon_n$, which is given below.
\begin{equation}
\Upsilon_n=\left\{\{P^{s1}_i\}, \{P^{s0}_i\}, P^{r}, \{\alpha_i\}, \{\beta_i\}, \{\gamma_i\} \right\}_n.
\end{equation}
The objective function for a Nash bargaining game is defined as a Nash bargaining solution, which can satisfy all the players in the game. In a Nash bargaining solution, all players reach a consensus of their strategies to maximize the product of utilities of all players\cite{bargaining-1999}. Mathematically, the Nash bargaining solution in this paper is defined as follows.
\begin{equation}
\begin{aligned}
&\max_{\{P^{s0}_i\},\{P^{s1}_i\},P^r,\{\alpha_i\},\{\beta_i\},\{\gamma_i\}}\Phi=\left(\prod_{(s_i,d_i)}{U^S_i}\right)U^R,\\
&~~~~~~~~~~~~~~~s.t.~U^S_i>0,~U^R>0,\\
&~~~~~~~~~~~~~~~~~~~~\Upsilon_1=\Upsilon_2=\ldots=\Upsilon_N=\Upsilon_{N+1}\\
&~~~~~~~~~~~~~~~~~~~~=\left\{\{P^{s1}_i\}, \{P^{s0}_i\}, P^{r}, \{\alpha_i\}, \{\beta_i\}, \{\gamma_i\} \right\},\\
&~~~~~~~~~~~~~~~~~~~~\mbox{Constraints~in}~(\ref{equ:timesplitting}),(\ref{equ:timeconstraint}),(\ref{equ:powerconstraint}),
\end{aligned}
\label{equ:obj}
\end{equation}
where $U^S_i$ and $U^R$ are given in~(\ref{equ:sourceutility}) and~(\ref{equ:relayutiliy}), respectively.\par
We emphasize that the main concerning in this paper is providing some static insights in a Nash bargaining solution in the wireless powered relay networks instead of designing a dynamic distributed process. Therefore, we assume that throughout this paper, the bargaining process is simulated in a centralized way on a randomly selected source. After the selected source achieves the values of $\{P^{s0}_i\},\{P^{s1}_i\},P^r,\{\alpha_i\},\{\beta_i\}$ and $\{\gamma_i\}$ in a Nash bargaining solution, it sends the values of $P^{s0}_i$ and $P^{s1}_i$  to each source $s_i$ through the relay to tell $s_i$ how to set the transmission power. In addition, the relay also achieve the the values of $\alpha_i$, $\beta_i$, and $\gamma_i$ from the selected source and sends them to each source $s_i$ to tell $s_i$ how to allocate the time between the energy transmission and data transmission.
\par
\section{Problem Analysis}
\label{sec:problemanalysis}
In this section, we first highlight that the Nash bargaining problem in this paper is not guaranteed to be a concave~(convex) problem, which is hard to obtain a strictly global optimal solution. Therefore, to find a suboptimal solution, we decompose the problem into three parts:~1)~energy transmission power optimization,~i.e.,~how to set $P_i^{s1}$; 2)~data transmission power optimization,~i.e.,~how to set $P_i^{s0}$ and $P^r$; and~3)~time division optimization,~i.e.,~how to set $\alpha_i$, $\beta_i$ and $\gamma_i$. We find this decomposition is not only traditionally meaningful,\footnote{When the scenario is complicated, it is a traditional trend of thoughts to decompose the time division and power control problem in dealing with resource allocation problems in wireless communication.} but also mathematically tractable that each problem holds its specific properties which can help simplify the process of the algorithm design to find a suboptimal result of the problem in~(\ref{equ:obj}).\par
\subsection{Non-Concavity of The Original Problem}
In this subsection, we prove that the original problem is not a concave~(convex) problem~(Property~\ref{the:notconcave}), which is hard to find a strictly optimal solution~\cite{convex-2004}.\par
\begin{lemma}
The objective function in (\ref{equ:obj}) is maximized, if and only if the following equation is satisfied for each source-destination pair:
\begin{equation}
\beta_i \log\left(1+\frac{P_i^{s0}|h_i|^2}{\sigma^2}\right)=\gamma_i \log\left(1+\frac{P^r|g_i|^2}{\sigma^2}\right).
\label{equ:onemanydatatrans}
\end{equation}
That is to say, at the Nash bargaining solution, the relay transmits all the data which it receives from its assisting sources.
\label{lem:datatrans}
\end{lemma}
\begin{proof}
When ${\beta_i \log\left(1+\frac{P_i^{s0}|h_i|^2}{\sigma^2}\right)<\gamma_i \log\left(1+\frac{P^r|g_i|^2}{\sigma^2}\right)}$, the relay can lower down its transmission time $\gamma_iT$, inducing a larger ${\left(\eta\sum_{x_i=1}(\alpha_i T P^{s1}_i|h_i|^2)-(\sum_{x_i=1}\gamma_i T P^r+E_0)\right)}$ in (\ref{equ:relayutiliy}). On the other hand, if~${\beta_i \log\left(1+\frac{P_i^{s0}|h_i|^2}{\sigma^2}\right)>\gamma_i \log\left(1+\frac{P^r|g_i|^2}{\sigma^2}\right)}$, the source can lower down its transmission time $\beta_iT$, inducing a larger $1/(\alpha_i P^{s1}+\beta_i P^{s0})$, which increases the objective function. Overall, the objective function is maximized, if and only if $\beta_i \log\left(1+\frac{P_i^{s0}|h_i|^2}{\sigma^2}\right)=\gamma_i \log\left(1+\frac{P^r|g_i|^2}{\sigma^2}\right)$ is satisfied.
\end{proof}
\begin{property}
The problem  (\ref{equ:obj}) is not guaranteed to be a concave~(convex) optimization problem with variables~ $\{\alpha_i\},\{\beta_i\},\{\gamma_i\},\{P^{s0}_i\},\{P^{s0}_i\}$ and $P^{r}$.
\label{the:notconcave}
\end{property}
\begin{proof}
Based on Lemma~\ref{lem:datatrans} and the power constraint $P^r\leq P^{r}_0$ as mentioned in Section II.C, we have the constraint set as follows:
\begin{equation}
\frac{\sigma^2}{|g_i|^2}\left(\left(1+\frac{P^{s0}_i|h_i|^2}{\sigma^2}\right)^{\frac{\beta_i}{\gamma_i}}-1\right)\leq P^{r}_0.
\end{equation}
It is not a convex set for $P^{s0}_i$ and $\beta_i$. Hence, the problem in (\ref{equ:obj}) is not a strict concave~(convex) optimization problem~\cite{convex-2004}.
\end{proof}
\subsection{Dedicators and Enjoyers for Energy Transmission}
In this subsection, we prove that at the Nash bargaining solution, the sources can be divided into two groups as \emph{dedicators} and \emph{enjoyers}. The definitions of the dedicator and the enjoyer are given as follows:
\begin{definition}
Source $s_i$ is defined as a \emph{dedicator} when source~$s_i$ transmits energy with full transmission power,~i.e.,~ $P_i^{s1}=P_0^{s}$.
\end{definition}
\begin{definition}
Source $s_i$ is defined as an \emph{enjoyer} when source~$s_i$ never transmits energy to the relay $r$,~i.e.,~$P_i^{s1}=0$.
\end{definition}
\par
Then, we have the theorem as follows:
\begin{theorem}
The objective function in (\ref{equ:obj}) is maximized if and only if $P_i^{s1}=P_0^s$ or $P_i^{s1}=0$ for each source $s_i$. That is to say, each source in the system is either a \emph{dedicator} or an \emph{enjoyer}.
\label{the:dedicator}
\end{theorem}
\begin{proof}
See Appendix \ref{app:dedicator}.
\end{proof}
\par
Then, we present Corollary~\ref{cor:ehtime} and Corollary~\ref{cor:ehorder}, respectively, to emphasize the energy transmission time property for the dedicators and which sources become more likely to be dedicators. These two corollaries can help simplify the process to decide the energy transmission time and the dedicator selection, which is shown in the algorithm in Section~\ref{sec:algorithm}.\par
\begin{corollary}
For each dedicator $s_i$~(with $P_i^{s1}=P_0^{s}$), its energy transmission time is $\alpha_iT=\max_{j=1}^{N}\{\alpha_jT\}$. That is to say, all dedicators spend the same time duration for energy transmission.
\label{cor:ehtime}
\end{corollary}
\begin{proof}
According to Theorem~\ref{the:dedicator}, for $P_i^{s1}=P_0^{s}$, function~(\ref{equ:ehmonotone}) is a monotonically increasing function for $P_i^{s1}$. Under this condition, it is a monotonically increasing function with respect to $\alpha_i$. Thus, each source $s_i$ tends to maximize~$\alpha_i$ in maximizing the objective function~(\ref{equ:obj}). Hence, we have~$\alpha_i=\max_{j=1}^{N}\{\alpha_j\}$.
\end{proof}

\begin{corollary}
For each pair of sources $s_i$ and $s_j$, $P_{i}^{s1}\geq P_{j}^{s1}$ is satisfied if and only if $\beta_{i} P_i^{s0}\geq \beta_{j} P_j^{s0}$ and $|h_1|^2\geq|h_2|^2$. That is to say, the source is more likely to do energy transmission as a dedicator, when it has a better channel condition~($|h_i|^2$) to the relay and holds more time~($\beta_iT$) and larger power~($P_i^{s0}$) for data transmission.
\label{cor:ehorder}
\end{corollary}
\begin{proof}
See Appendix \ref{app:ehorder}.
\end{proof}
\par
Corollary~\ref{cor:ehorder} can also help to design an incentive mechanism, to encourage sources to be dedicators instead of receivers. According to Corollary~\ref{cor:ehorder}, in a certain time block, sources who want to transmit more data~(with a larger $\beta_iT$) are more likely to be selected as dedicators. Therefore, assuming that sources in the network have the similar amount of data in a relatively long time~(a number of time blocks), some sources have more data in the first several time blocks are selected as dedicators and the others are responsible to transmit energy in the rest ones. In addition, the channel statement in the wireless network usually changes with time. This indicates that sources as enjoyers are eventually selected as dedicators in the later time blocks, when their channel statements between the relay become better. Even though Corollary~\ref{cor:ehorder} can avoid the circumstance that some sources are always working as dedicators and others are always selected as enjoyers, more incentive mechanisms are needed to proactively encourage sources to work as dedicators, which is left for the future work.\par
\subsection{Quasi-Concavity of Data Transmission Power Control}
In this subsection, we prove that when the time division and the group of \emph{dedicators} are fixed, the data transmission power control is a quasi-concave optimization problem.
\begin{theorem}
Given fixed $\{P_i^{s1}\}$ and fixed $\{\alpha_i\},\{\beta_i\},\{\gamma_i\}$, the problem in (\ref{equ:obj}) is a quasi-concave optimization problem with variables $\{P^{s0}_i\} $ and $P^{r}$.\par
\label{the:quasipower}
\end{theorem}
\begin{proof}
See Appendix \ref{app:concavepower}.
\end{proof}
\par
\subsection{Quasi-Concavity of Time Division}
In this subsection, we prove that assuming the data transmission power control and the group of \emph{dedicators} to be fixed, the time division problem is a quasi-concave optimization problem. \par
\begin{theorem}
Given fixed $\{P_i^{s1}\}$ and fixed $\{P^{s0}_i\} $ and $P^{r}$, the problem in (\ref{equ:obj}) is a quasi-concave optimization problem with variables $\{\alpha_i\},\{\beta_i\},\{\gamma_i\}$.
\label{the:concavetime}
\end{theorem}
\begin{proof}
See Appendix \ref{app:concavetime}.
\end{proof}
\par
We conclude this section as follows. The original problem~in~(\ref{equ:obj}) is not a concave~(convex) problem which can not guarantee to obtain a globally optimal solution~(Property~\ref{the:notconcave}). However, when we decompose the problem into three parts:~1)~energy harvesting power optimization, 2)~data transmission power optimization, and 3)~time division, we find that the problem becomes simpler. The energy transmission power optimization is a discrete optimization problem. In particular, the sources can be divided into \emph{dedicators} and \emph{enjoyers}~(Theorem~\ref{the:dedicator}), which are relevant to the channel conditions from the sources to the relay and the transmission energy cost of the sources~(Corollary~\ref{cor:ehorder}). The data transmission power optimization problem and time division problem are both quasi-concave optimization problems~(Theorem~\ref{the:quasipower} and Theorem~\ref{the:concavetime}). Therefore, there exist many optimization algorithms to solve~\cite{convex-1970}\cite{convex-2004}, e.g.~the gradient ascent algorithm. The properties discussed in this section indicate that the problem decomposition is not only traditionally meaningful, but also reasonable from a mathematical perspective, which can guide us to design an efficient algorithm in the next section.\par

\section{Alternating Power Control and Time Division Algorithm}
\label{sec:algorithm}
\begin{algorithm}[t]
\caption{Alternating Power Control and Time Division Algorithm}
\KwIn{$k=0$, $(\{P_i^{s1}\})^L$, $(\{P_i^{s0}\})^k$,$(P^r)^k$,$(\{\alpha_i\},\{\beta_i\},\{\gamma_i\})^k$},
\KwOut{$\max\Phi(\{P_i^{s1}\},\{P_i^{s0}\},P^r,\{\alpha_i\},\{\beta_i\},\{\gamma_i\})$}
1)~Enumerating possible conditions of dedicator selection, at each condition with $\{P_i^{s1}\}^L$:\\
k=0\;
\While{Arbitral stopping criterion is not satisfied}
{
    (a)~Transmission~power~optimization:~$(\{P_i^{s0}\},P^r)^{k+1}=\arg\max_{(\{P_i^{s0}\},P^r)}\Phi_1^k$\;
    (b)~Time~division~optimization:~$(\{\alpha_i\},\{\beta_i\},\{\gamma_i\})^{k+1}=\arg\max_{\{\alpha_i\},\{\beta_i\},\{\gamma_i\}}\Phi_2^k$\;
}
$\Phi^L=\Phi^L\left(\{P_i^{s1}\}^L,(\{P_i^{s0}\},P^r,\{\alpha_i\},\{\beta_i\},\{\gamma_i\})^{k+1}\right)$\;
2)~$\Phi=\max_{L}\Phi^L$\;
*~$(\cdot)^L$ is the $L^{th}$ condition to select dedicators, $(\cdot)^{k}$ is the value of the $k^{th}$ iteration.
\label{alg:scheme}
\end{algorithm}
In this section, we design an alternating algorithm to solve the bargaining problem in~(\ref{equ:obj}). The analysis in Section~\ref{sec:problemanalysis} enlightens the decomposed algorithm design: 1)~The outer loop is optimizing the energy transmission power with enumerating all the possible conditions to select dedicators~(according to Theorem~\ref{the:dedicator}). 2)~Under each condition, we alternatively and iteratively solve data transmission power optimization and time division problems~(according to Theorem~\ref{the:quasipower} and Theorem~\ref{the:concavetime}). We emphasize that since the original problem is not a concave (convex) optimization problem~(according to Property~\ref{the:notconcave}), this kind of algorithm can only guarantee to reach a suboptimal solution.  Specifically, the alternating algorithm is shown in Algorithm~\ref{alg:scheme},~where $\Phi_1^k$ and $\Phi_2^k$ are given as follows:
\begin{equation}
\left\{
\begin{aligned}
&\Phi_1^k=\Phi(\{P_i^{s1}\}^L,\{P_i^{s0}\},P^r,\{\alpha_i\}^k,\{\beta_i\}^k,\{\gamma_i\}^k),\\
&\Phi_2^k=\Phi(\{P_i^{s1}\}^L,(\{P_i^{s0}\},P^r)^{k+1},\{\alpha_i\},\{\beta_i\},\{\gamma_i\}),\\
\end{aligned}
\right.
\label{equ:timeiter}
\end{equation}
where $(\cdot)^L$ denotes the value at the $L$'s condition of energy transmission power optimization, and $(\cdot)^k$ is the value at the $k$'s iteration of data transmission power optimization and time division. $\Phi_1^k$ equals to the objective $\Phi$ in~(\ref{equ:obj}) when $\{P_i^{s1}\}$ and $\{\alpha_i\},\{\beta_i\},\{\gamma_i\}$ are fixed. $\Phi_2^k$ equals to the objective $\Phi$ when $\{P_i^{s1}\},\{P_i^{s0}\},P^r$ are fixed.\par
More details on energy transmission power optimization~(dedicator selection) and the details on the alternating the data transmission power optimization and time division process are given as below. In addition, we also discuss the convergence and the complexity of our algorithm in this section.\par
\subsection{Energy Transmission Power Optimization}
In this subsection, we explain the dedicator selection in Algorithm~\ref{alg:scheme}. According to Theorem~\ref{the:dedicator}, the energy transmission power optimization is first transformed from a continuous problem into a discrete problem. Thus, what we concern about is to decide which sources are selected to transmit energy as dedicators and which are not. According to Theorem~\ref{the:dedicator} and Corollary~\ref{cor:ehtime}, we emphasize that the power and time for energy transmission for all dedicators are identical. For convenience, we define an indicator vector $\Omega=\{\omega_1,\omega_2,\ldots,\omega_N\}\in \mathbb{B}$, where $\omega_i=1$ when source $s_i$ is a dedicator. Otherwise $\omega_i=0$ and source $s_i$ is an enjoyer. $\mathbb{B}$ is the set of the all the conditions of dedicator selection, where $|\mathbb{B}|=2^N$. We provide Remark~\ref{rem:dedicator1} and Remark~\ref{rem:dedicator2} as below to help us decide which sources are more likely to be selected as dedicators.
\begin{remark}
Given any index $K$, when $\sum_{i=1}^N \omega_i=K$ is not feasible for all elements in $\mathbb{B}$, the conditions with $\sum_{i=1}^N \omega_i<K$ are not feasible as well.
\label{rem:dedicator1}
\end{remark}
Remark~\ref{rem:dedicator1} indicates that the transmitted energy is not enough with $K$ dedicators, it is also not enough when the number of dedicators is below $K$. Thus, we enumerate the conditions from $N$ dedicators to $N-1,N-2,\ldots$ dedicators and stop when the result is not feasible for all conditions with $K$ dedicators.
\begin{remark}
When $\sum_{i=1}^N \omega_i=K$, $\omega_1 \geq \omega_2$ is \emph{possibly} satisfied when $|h_i|^2\geq|h_j|^2$.
\label{rem:dedicator2}
\end{remark}
According to Corollary~\ref{cor:ehorder}, we have that the energy transmission power $P_{i}^{s1}\geq P_{j}^{s1}$ is satisfied if and only if $\beta_{i} P_i^{s0}\geq \beta_{j} P_j^{s0}$ and $|h_1|^2\geq|h_2|^2$. However, $\beta_{i}$ and $P_i^{s0}$ for each source $s_i$ cannot be determined before the alternating data transmission power and time division process is accomplished in Algorithm~\ref{alg:scheme}. Thus, we can only decide $\omega_i$ based on $|h_i|^2$, which results in that the operation in Remark~\ref{rem:dedicator2} might miss the optimal point. However, the operation can observably reduce the enumeration time with $K$ dedicators, from $\binom{N}{K}$ to $1$ ~(to select $K$ sources as dedicators with the highest $|h_i|^2$).
\par
%We conclude the dedicator selection strategy as follows. We propose two remarks~(Remark~\ref{rem:dedicator1} and Remark~\ref{rem:dedicator2}) to help us reduce the time of enumeration for dedicator selection. Even though these two reductions are not strict, they can also help us eliminate a large number of conditions that less possibly exist in the optimal set.\par
\subsection{Data Transmission Power Optimization and Time Division}
In this subsection, we propose the concave~(convex) optimization methods to deal with the data transmission power optimization and time division, i.e., how to find the optimal $\{P_i^{s0}\}$, $P^r$, $\{\alpha_i\}$, $\{\beta_i\}$, and $\{\gamma_i\}$. We decompose the data transmission power optimization and time division and alternatively deal with the two problems, i.e., the inner iteration in Algorithm~\ref{alg:scheme}, according to Property~\ref{the:notconcave}, Theorem~\ref{the:quasipower} and Theorem~\ref{the:concavetime}. The specific iteration of data transmission power optimization and time division, and the arbitral stopping criterion are given as follows.
\begin{itemize}
\item \textbf{Data Transmission Power Optimization:} When the energy transmission power $\{P_i^{s1}\}$ and the time division $~\{\alpha_i\},~\{\beta_i\},~\{\gamma_i\}$ are fixed as the values obtained from the latest iteration, the data transmission optimization problem at the $k^{th}$ iteration is to maximize $\Phi_1^k$ in Algorithm~\ref{alg:scheme}, which is given as follows:
    \begin{equation}
    \small
    \begin{aligned}
    &\max_{P^r} \Phi_1^k\\
    &~~~~~~=\prod_{i=1}^{N}\left(\frac{\log(1+L_1^iP^r)}{L_3^i((1+L_1^iP^r)^{L_2^i}-1)+L_4^i}\right)\cdot \left(K_1-K_2P^r\right),\\
    &~~s.t.~0\leq P^r\leq [P_0^r]^{+},\\
    \end{aligned}
    \label{equ:quasipowermain}
    \end{equation}
    where $[P_0^r]^{+}=\min\left\{P_0^r,\{\frac{\sigma^2}{|g_i|^2}(1+\frac{P_0^s|h_i|^2}{\sigma^2})^{\frac{\beta_i}{\gamma_i}}-\frac{\sigma^2}{|g_i|^2}\}\right\}$ and $L_1^i$ to $L_4^i$ and $K_1$ to $K_2$ are given in (\ref{equ:proofpoweralgebra}) in Appendix~\ref{app:concavepower}. How to transfer the original problem in~(\ref{equ:obj}) into the above form is also given in Appendix~\ref{app:concavepower}. The Lagrangian function $\mathcal{L}_1^k(P^r,\rho_1,\rho_2)$ is given as follows:
    \begin{equation}\small
    \begin{aligned}
    &\mathcal{L}_1^k(P^r,\rho_1,\rho_2)\\
    &=-\prod_{i=1}^{N}\left(\frac{\log(1+L_1^iP^r)}{L_3^i((1+L_1^iP^r)^{L_2^i}-1)+L_4^i}\right)\left(K_1-K_2P^r\right)\\
    &+\rho_1(0-P^r)+\rho_2(P^r-[P_0^r]^*),\\
    \end{aligned}
    \label{equ:lagpower}
    \end{equation}
    where $\rho_1>0$ and $\rho_2>0$ are the Lagrange multipliers. Then, $\mathcal{L}_1^k(P^r,\rho_1,\rho_2)$ is minimized through the typical gradient ascent algorithm~(Chapter 3 of~\cite{convex-2004}). The details of the gradient ascent algorithm for data transmission power optimization is given in Appendix~\ref{app:poweropt}. We obtain the optimal data transmission power of the relay as $(P^r)^{\star}$ and calculate the optimal data transmission power of the sources as $(\{(P^{s0}_i)\})^{\star}$ based on~(\ref{equ:onemanydatatrans}). Finally, we update $P^r$ and $\{P^{s0}_i\}$ as $(P^r)^{k+1}=(P^r)^{\star}$ and $(\{P^{s0}_i\})^{k+1}=(\{P^{s0}_i\})^{\star}$.
\item \textbf{Time Division:} When the energy transmission power $\{P_i^{s1}\}$ are fixed and the data transmission power $P^r$ and $\{P^{s0}_i\}$ have already been updated at the $k^{th}$ iteration, the time division problem at the $k^{th}$ iteration is to maximize $\Phi_2^k$ in Algorithm~\ref{alg:scheme}, which is given as follows:
    \begin{equation}\small
    \begin{aligned}
    &\max_{\{\alpha_i\},\{\gamma_i\}} \Phi_2^k\\
    &~~~~~~=\left(\prod_{P_i^{s1}=P_0^s} \frac{D_1^i\gamma_i}{D_2^i\alpha+D_3^i\gamma_i}\right)\cdot(F_1\alpha-F_2-\sum_{i=1}^ND_4^i\gamma_i),\\
    &~~s.t.~\alpha+\sum_{i=1}^N(\gamma_i+D_5^i\gamma_i)=1,\\
    &~~~~~~~(\gamma_i+D_5^i\gamma_i)T\geq\theta_0T,\\
    \end{aligned}
    \label{equ:quasitimemain}
    \end{equation}
    where $D_1$ to $D_5$ and $F_1$ to $F_2$ are given in~(\ref{equ:prooftimealgebra}) in Appendix~\ref{app:concavetime}. How to transform the original problem in (\ref{equ:obj}) into the above form is also given in Appendix~\ref{app:concavetime}. The Lagrangian function $\mathcal{L}_2^k(\alpha,\{\gamma_i\},\varrho_1,\{\varrho_2^i\})$ is given as follows:
    \begin{equation}\small
    \begin{aligned}
    &\mathcal{L}_2^k(\alpha,\{\gamma_i\},\varrho_1,\{\varrho_2^i\})\\
    &=-\left(\prod_{P_i^{s1}=P_0^s} \frac{D_1^i\gamma_i}{D_2^i\alpha+D_3^i\gamma_i}\right)\cdot(F_1\alpha-F_2-\sum_{i=1}^ND_4^i\gamma_i)\\
    &+\varrho_1\left(\alpha+\sum_{i=1}^N(\gamma_i+D_5^i\gamma_i)-1\right)\\
    &+\sum_{i=1}^N\varrho_2^i\left(\theta_0T-(\gamma_i+D_5^i\gamma_i)T\right),\\
    \end{aligned}
    \label{equ:lagtime}
    \end{equation}
    where $\varrho_1$ and $\{\varrho_2^i\}$ are the Lagrange multipliers. Similar to the data transmission power optimization, $\mathcal{L}_2^k(\alpha,\{\gamma_i\},\varrho_1,\{\varrho_2^i\})$ is minimized through the typical gradient ascent algorithm~(Chapter 3 of~\cite{convex-2004} and find the optimal point as $\alpha^\star$, $(\{\gamma_i\})^{\star}$ and $(\{\beta_i\})^{\star}$. More details of the dual ascent algorithm for time division is given in Appendix~\ref{app:timeopt}. Finally we update the time division as $\alpha^{k+1}=\alpha^{\star}$, $\{\gamma_i\}^{k+1}=\{\gamma_i\}^{\star}$ and $\{\beta_i\}^{k+1}=\{\beta_i\}^{\star}$.
\item \textbf{Stopping Criterion:} The stopping criterion in the optimization methods is used to determine when the iteration stops~\cite{convex-2004}. Specifically, in Algorithm~\ref{alg:scheme} the stopping criterion is the rule to evaluate when the alternating data transmission power optimization and the time division process should stop, which is given as follows,
    \begin{equation}
    \max\left\{\frac{|\Phi_1^{k+1}-\Phi_1^{k}|}{\Phi_1^{k}},~\frac{|\Phi_2^{k+1}-\Phi_2^{k}|}{\Phi_2^{k}}\right\}<\varepsilon,
    \label{equ:stoppingcriterion}
    \end{equation}
    where $\varepsilon$ is a minor parameter to determine the convergence.
\end{itemize}
\par
\subsection{Convergence and Complexity}
In this part, we discuss the convergence and the complexity of our proposed algorithm. Specifically, we at first prove that the alternating data transmission power optimization and time division process (the inner loop in Algorithm~\ref{alg:scheme}). Then, we discuss the complexity of our algorithm, i.e., to discuss the enumeration time of energy transmission power optimization (the outer loop in Algorithm~\ref{alg:scheme}) and the iteration time of alternating data transmission power optimization and time division process~(the inner loop in Algorithm~\ref{alg:scheme}). \par
We at first prove the convergence of this alternating data transmission power optimization and time division process.
\begin{theorem}
The alternating data transmission power optimization and time division process converges.
\label{the:convergence}
\end{theorem}
\begin{proof}
It is easy to test that the singularity of the objective function $\Phi$ in~(\ref{equ:obj}) exists, only when for any source $s_i$, $P_i^{s0}=0$ and $P_i^{s1}=0$ simultaneously happen. However, $P_i^{s0}=0$ indicates that $s_i$ does not transmit the data~($\beta_i=0$), which contradicts the constraints in ($\ref{equ:timeconstraint}$). Therefore, $\Phi$ has an upper bound in the feasible zone. The Lagrangian methods in the alternating data transmission power optimization and time division process guarantee that $\Phi_1^k$ and $\Phi_2^k$ are non-decreasing at each iteration~\cite{convex-2004}. When the objective function $\Phi$ keeps increasing and has an upperbound, the stopping criterion will finally become satisfied with the iteration execution~\cite{convex-2004}, which indicates that the iteration converges.
\end{proof}
\par
We then discuss the complexity~(the total number of iterations) of our algorithm. Based on Corollary~\ref{cor:ehorder}, Remark~\ref{rem:dedicator1} and Remark~\ref{rem:dedicator2}, the enumeration time of the outer loop is prominently reduced from $2^{N}$ to $N$. Thus, the computation complexity can be expressed as $N\times T_{in}$, where $T_{in}$ is the average number of iterations for the inner loop. As for the inner loop of our algorithm, either for the data transmission power optimization or for the time division, the general gradient ascent method usually guarantees a linear convergence speed~\cite{numerical-2006}. We use $O(\cdot)$ to represent the magnitude of iteration time. Then, for the data transmission power optimization, the iteration time can be expressed as~$O(\frac{1}{\varepsilon_1})$, where $\varepsilon_1$ is the parameter to constrain the accuracy of the data transmission power optimization as given in~(\ref{equ:poweraccuracy}) in Appendix~\ref{app:poweropt}. For the time division, the iteration time can be expressed as $O(\frac{1}{\varepsilon_2})$, where $\varepsilon_2$ represents the accuracy of time division as given in~(\ref{equ:timeaccuracy}) in Appendix~\ref{app:timeopt}. Thus, we can calculate the computation complexity, denoted by $T$, as follows.
\begin{equation}\small
T=N\times T_{in}= N\times N_{al} \times O(\frac{1}{\varepsilon_1}+\frac{1}{\varepsilon_2}),
\end{equation}
where $N_{al}$ is the time of alternations between the data transmission power optimization and the time division. We emphasize that even though the alternating data transmission power optimization and the time division process can be proved to converge, the strict closed-form relation between the $N_{al}$ and the arbitral stopping criterion given in~(\ref{equ:stoppingcriterion}) is complicated to derive. This is because the non-convexity~(non-concavity) of the original problem results in the distance between the point after each step of data transmission power optimization or time division and the final result is hardly to calculate. Thus, we left the mathematical derivation for the future work.
\par
We conclude this section as follows. We design an alternating power control and time division algorithm to solve the bargaining problem in~(\ref{equ:obj}). The outer loop in the Algorithm~\ref{alg:scheme} is the energy transmission power optimization. Properties in Remark~\ref{rem:dedicator1} and Remark~\ref{rem:dedicator2} can help to reduce the time of enumeration. The inner iteration in the Algorithm~\ref{alg:scheme} is the data transmission power optimization and time division. At each iteration, we sequentially optimize the data transmission power and deal with the time division, through the Lagrangian optimization methods. The stopping criterion guarantees the convergence of the inner iteration. \par

\section{Numerical Results}
In this section, we provide the numerical results to illustrate the properties of the problem and our alternating algorithm. We at first show the convergence of our algorithm in Fig.~\ref{simfig:convergence}. Then, we discuss the relation between objective in~(\ref{equ:obj}) and the number of dedicators in Fig.~\ref{simfig:numberofdedicators}. The sum capacity of all source-destination pairs and the residual harvested energy in the relay are illustrated in Fig.~\ref{simfig:sumcapacity} and Fig.~\ref{simfig:residualenergy}. The distribution of the utilities among dedicators and enjoyers is shown in Fig.~\ref{simfig:distributionsource}, which shows the necessity to design incentive mechanisms in the future work. The simulation parameters are given in Table~\ref{sim:simulationparameter}.\par
\begin{table*}[!t]
\renewcommand{\arraystretch}{1.0}
\caption{Simulation Parameters}\label{sim:simulationparameter} \centering
\begin{tabular}{{l|l|l}}
\hline
\hline
Network&Region Size~($L\times W$)& Constant $20$m$\times$$20$m\\
Compositions&Number of Source-Destination Pairs~($N$)& Constant $7$~~\\
~~&Distribution of Sources~~&Uniform Distribution in Square Region\\
~~&&~~~~~~~from $(0,0)$ to $(0.5L,0.5W)$~~\\
~~&Distribution of Destinations~~&Uniform Distribution in Square Region\\
~~&&~~~~~~~from $(0.5L,0.5W)$ to $(L, W)$~~\\
~~&Location of The Relay~~&Constant~$(0.5L,0.5W)$~~\\
\hline
Data~~&The Highest Transmission Power of Each Source~($P_{0}^s$)~~&Constant 10mW~~\\
And~~&The Highest Transmission Power of The Relay~($P^r_0$)~~&Constant 10mW~~\\
Energy~~&Noise Power~($n^2$)~~&Constant -95dBm~~\\
Transmission~~&Channel Gain~($\gamma$)~~&Average Passloss: Distance$^{-2}$~~\\
~~&Circuit Energy Cost Per Time Slot of The Relay~($\frac{E_0}{T}$)~~&Variable from $0$ to $0.2$mW~~\\
~~&Energy Harvesting Efficiency of The Relay~($\eta$)~~& Constant 0.5~~\\
~~&Time Division Parameter~($\theta_0$) &Constant 0.05~~\\
\hline
Algorithm&Convergence Decision Parameter~($\epsilon$)& Constant $10^{-4}$~~\\
Parameter&&\\
\hline
\hline
\end{tabular}
\end{table*}

\begin{figure}[!t]
\centering
\includegraphics[width=3.5in]{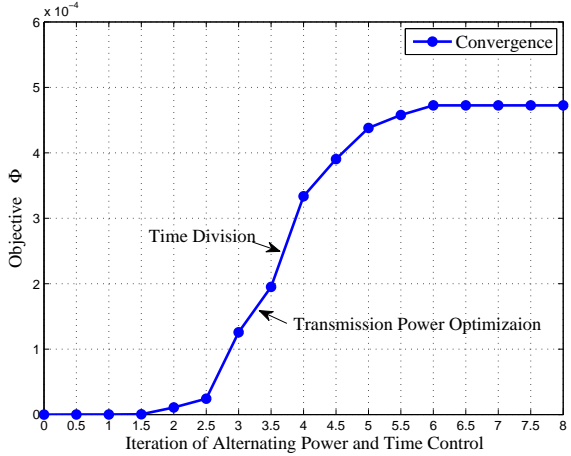}
\caption{Convergence of the iteration for data transmission power optimization and time division.}
\label{simfig:convergence}
\end{figure}
In Fig.~\ref{simfig:convergence}, we illustrate the convergence of the alternating data transmission power optimization and time division presented in Algorithm~\ref{alg:scheme}. To clearly show the process of the algorithm, we assume that all the sources are dedicators. At each iteration, the first part~(from $t$ to $t.5$) is the updating of the data transmission power~($P^r$ and $\{P^{s0}_i\}$), and the second part~(from $t.5$ to $t+1$) is the updating of the time division. We can find that with the iteration of alternating algorithm, the objective $\Phi$ keeps increasing and finally converges, even though the original problem is neither strictly convex nor strictly concave. In addition, it can be observed that the data transmission power optimization and the time division steps improve the objective function in the similar order of magnitude. This indicates that it is hardly to say which step has a stronger influence on the objective optimization.\par
\begin{figure}[!t]
\centering
\includegraphics[width=3.5in]{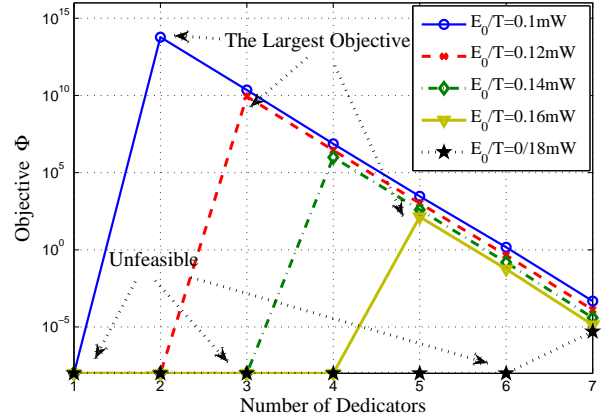}
\caption{The relation between the objective and the number of dedicators.}
\label{simfig:numberofdedicators}
\end{figure}
In Fig.~\ref{simfig:numberofdedicators}, we show the relation between the objective~($\Phi$) and the number of dedicators, when the energy for relay's decoding and recoding~($E_0$) varies. With each fixed $E_0$, we show a good representative condition for intuitive illustration according to Remark~\ref{rem:dedicator1} and Remark~\ref{rem:dedicator2}. We can find that with fixed $E_0$, the objective~$\Phi$ is infeasible when the number of dedicators is small. This is because when the number of dedicators is not enough, the transmitted energy is not enough for the relay's data transmission. However, when the number of dedicators is sufficient, less dedicators leads to a larger $\Phi$,~e.g.,~with 3 dedicators when $E_0$ equals $0.12$mW$\cdot T$. \par
\begin{figure}[!t]
\centering
\includegraphics[width=3.5in]{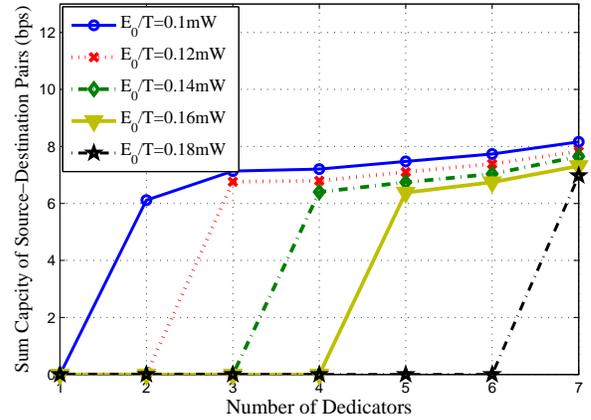}
\caption{The relation between the sum capacity and the number of dedicators.}
\label{simfig:sumcapacity}
\end{figure}
In Fig.~\ref{simfig:sumcapacity}, we show the relation between the sum capacity~($\sum_{i=1}^N C_i$) of the source-destination pairs and the number of dedicators, when the energy for relay's decoding and recoding~($E_0$) varies. We can find that with more sources acting as dedicators, the sum capacity keeps increasing. This is because with more dedicators, the relay harvests more energy with a shorter harvesting duration~($\alpha T$). Correspondingly, the relay has a longer time ($\sum_{i=1}^N(\beta_i+\gamma_i)T$) to forward data. Thus, the sum capacity will be larger. This also illustrates the TSR protocol in Fig.~\ref{fig:dftsr}, where it is better that all sources work as dedicators to take full use of the time resources~($\alpha T$), when we tend to maximize the sum capacity from a centralized perspective.\par
\begin{figure}[!t]
\centering
\includegraphics[width=3.5in]{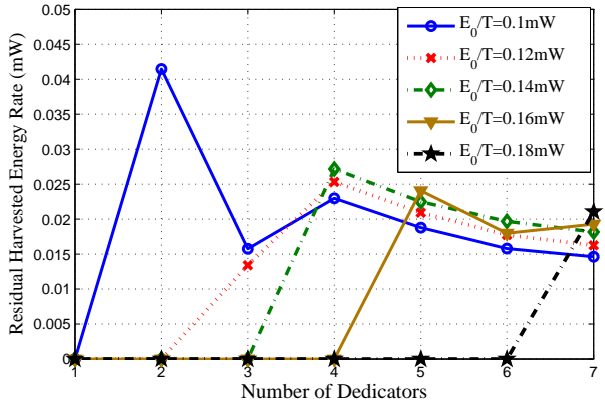}
\caption{The relation between the residual harvested energy in the relay and the number of dedicators.}
\label{simfig:residualenergy}
\end{figure}
In Fig.~\ref{simfig:residualenergy}, we show the relation between the residual harvested energy~($U_r$) in the relay and the number of dedicators. The residual harvested energy $U_r$ equals to zero when number of dedicators is not enough, which results in the harvested energy in the relay is not enough for data forwarding.
With more sources as dedicators, the residual harvested energy is not guaranteed to be increasing. This indicates that in the Nash bargaining game, to balance the utilities among sources and the relay, with more dedicators the sources require the relay to forward more data where may result in less residual harvested energy in the relay. \par
\begin{figure}[!t]
\centering
\includegraphics[width=3.5in]{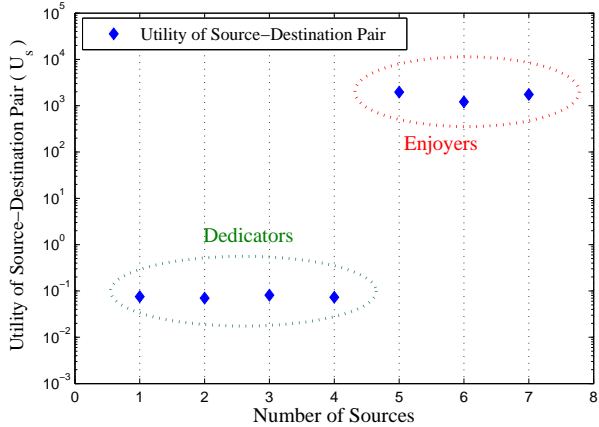}
\caption{Utility distribution among source-destination pairs.}
\label{simfig:distributionsource}
\end{figure}
Fig.~\ref{simfig:distributionsource} illustrates the utility distribution among source-destination pairs, with sources as dedicators and sources as enjoyers. We observe that the dedicators' utilities~($10^{-1}$) are dramatically lower than those of enjoyers~($10^3$). The possible reasons are given as follows. To guarantee quality of service of each source-destination pair's data transmission, it is reasonable to guarantee that each pair has a minimum duration in data transmission~($\theta_0T$) as shown in~(\ref{equ:timeconstraint}). However, the constraints in~(\ref{equ:timeconstraint}) possibly result in the utility imbalance among dedicators and enjoyers, where all sources tend to be enjoyers instead of dedicators. This indicates that the Nash bargaining solution can guarantee fairness only to a certain extent in one time block. This fairness problem will subside and diminish in the long term, since Corollary~\ref{cor:ehorder} and Remark~\ref{rem:dedicator2} result in that the all sources tend to be dedicators alternatively with their transmission channels changes. For example, when source $s_i$ has a relatively larger $|h_i|^2$, it will act as a dedicator. On the contrary, source $s_j$ will be a dedicator when having a relatively larger $|h_j|^2$ in the next time block. That is to say, enjoyers will eventually become dedicators according to Corollary~\ref{cor:ehorder} and Remark~\ref{rem:dedicator2}. However, more effective incentive mechanisms still need to be designed to guarantee fair utility distribution for more general cases, such as when some sources always act as dedicators if their channel conditions~($|h_i|^2$) are always better than others. These points are left for the future work.\par

\section{Conclusion}
In this paper, we have considered a wireless powered relay network with one relay and multiple source-destination pairs from a Nash bargaining game theoretical perspective. We have found that the problem is not a strictly concave~(convex) optimization problem, which is hard to solve. However, the problem has been proved to be simplified when it is decomposed into three parts: 1)~energy transmission power optimization, 2)~data transmission power optimization and 3)~time division. We have proved that with the Nash bargaining solution, the sources can be divided into two groups as dedicators, who do energy harvesting with the maximum power and enjoyers who do not transmit energy to the relay at all. In addition, we have proved the quasi-concavity of both transmission power control problem and transmission power control problem. Based on the analysis on the problem, we have designed an algorithm to find a suboptimal Nash bargaining solution. The convergence of the algorithm and the properties have been illustrated in the simulation results. In the simulation results, we have also shown the imbalance of utilities between dedicators and enjoyers at the Nash bargaining solution, which indicates that more effective incentive mechanisms are needed.\par

\appendices
\section{Proof of Theorem \ref{the:dedicator}}
\label{app:dedicator}
When we separately consider transmission power of energy harvesting~$P_i^{s1}$ for each source $s_i$, the problem~in~(\ref{equ:obj}) can be rewritten as a fractional programming problem, which is given as follows:
\begin{equation}
\begin{aligned}
&\max \frac{A_4+A_5(\alpha_iP_i^{s1})}{A_3(A_1\alpha_iP_i^{s1}+A_2)},\\
&~~s.t.~0\leq P_i^{s1} \leq P_0^{s},\\
\end{aligned}
\label{equ:ehmonotone}
\end{equation}
where $A_1$ to $A_5$  are independent from $P_i^{s1}$, which is given as follows:
\begin{equation}\small
\left\{
\begin{aligned}
&A_1=T,\\
&A_2=P_i^{s0}\beta_i,\\
&A_3=\frac{\prod_{k\neq i}P_k^{s1}(P_k^{s1}\alpha_kT+P_i^{s0}\beta_kT)}{\prod_{1\leq k\leq N}C_i},\\
&A_4=\eta \sum_{k\neq i}P_k^{s1}\alpha_kT-E_0.\\
&A_5=T,\\
\end{aligned}
\right.
\end{equation}
The objective function above is a monotonic function of $P_{i}^{s1}$. When $A_4<A_2$, the objective function is a monotonically increasing function of $P_i^{s1}$. The objective function is maximized only when $P_i^{s1}=P_0^{s1}$. On the other hand, when $A_4\geq A_2$, the objective function is a monotonically decreasing function of $P_i^{s1}$. The objective function is maximized only when $P_i^{s1}=0$. Hence, we have $P_{i}^{s1}$ equal to either $P_0^{s}$ or $0$ to maximize the objective in (\ref{equ:obj}). That is to say, the sources can be classified into two groups: the group of dedicators~with $P_i^{s1}=P_0^{s1}$ and the group of enjoyers~with $P_i^{s1}=0$.\par

\section{Proof of Corollary \ref{cor:ehorder}}
\label{app:ehorder}
When we set $P_i^{s1}=P_0^s$ and $P_j^{s1}=0$, the objective function~(\ref{equ:obj}) can be written as follows:
\begin{equation}
\Phi_i=\frac{B_2\alpha_i|h_i|^2P_0^s+B_1}{\beta_jP_j^{s0}(\alpha_iP_0^s+\beta_iP_i^{s0})},
\end{equation}
where $B_2>0$ and~$B_2\alpha_i|h_i|^2P_0^s+B_1>0$ according to the constraints in (\ref{equ:obj}). On the other hand, we can rewrite~(\ref{equ:obj}) as~$\Phi_j=\frac{B_2\alpha_j|h_j|^2P_0^s+B_1}{\beta_iP_i^{s0}(\alpha_jP_0^s+\beta_jP_j^{s0})}$ when setting $P_i^{s1}=0$ and $P_j^{s1}=P_0^{s1}$. In addition, according to Corollary~\ref{cor:ehtime}, we have $\alpha=\alpha_i=\alpha_j$. Then, when $\beta_i P_i^{s0}\geq \beta_j P_j^{s0}$ and $|h_i|^2\geq|h_j|^2$, we have the inequation as follows:
\begin{equation}\small
\begin{aligned}
&\Phi_i-\Phi_j\\
&=\frac{\beta_iP_i^{s0}\alpha P_0^s(B_2|h_i|^2P_0^s+B_1)-\beta_jP_j^{s0}\alpha P_0^s(B_2|h_j|^2P_0^s+B_1)}{\beta_iP_i^{s0}\beta_jP_j^{s0}(\alpha P_0^s+\beta_jP_j^{s0})(\alpha P_0^s+\beta_iP_i^{s0})}\\
&-\frac{\beta_iP_i^{s0}\beta_jP_j^{s0}\alpha P_0^{s}(|h_i|^2-|h_j|^2)}{\beta_iP_i^{s0}\beta_jP_j^{s0}(\alpha P_0^s+\beta_jP_j^{s0})(\alpha P_0^s+\beta_iP_i^{s0})}\\
&\geq\frac{\beta_iP_i^{s0}\alpha P_0^s(B_2|h_i|^2P_0^s+B_1)-\beta_jP_j^{s0}\alpha P_0^s(B_2|h_j|^2P_0^s+B_1)}{\beta_iP_i^{s0}\beta_jP_j^{s0}(\alpha P_0^s+\beta_jP_j^{s0})(\alpha P_0^s+\beta_iP_i^{s0})}\\
&\geq\frac{\beta_jP_j^{s0}\alpha P_0^s(B_2|h_i|^2P_0^s+B_1)-\beta_jP_j^{s0}\alpha P_0^s(B_2|h_j|^2P_0^s+B_1)}{\beta_iP_i^{s0}\beta_jP_j^{s0}(\alpha P_0^s+\beta_jP_j^{s0})(\alpha P_0^s+\beta_iP_i^{s0})}\\
&=\frac{\beta_j(P_j^{s0})^2\alpha P_0^sB_2(|h_i|^2-|h_j|^2)}{\beta_iP_i^{s0}\beta_jP_j^{s0}(\alpha P_0^s+\beta_jP_j^{s0})(\alpha P_0^s+\beta_iP_i^{s0})}\geq0.\\
\end{aligned}
\end{equation}
Therefore, we have the source be more likely to perform energy transmission as a dedicator~(as $s_i$ to~$s_j$), when it has better channel condition~($|h_i|^2$) to the relay, holds more time~($\beta_iT$) and larger power~($P_i^{s0}$) for data transmission.
\section{Proof of Theorem \ref{the:quasipower}}
\label{app:concavepower}
According to Lemma~\ref{lem:datatrans}, we can rewrite the problem~in~(\ref{equ:obj}) with variable $P^r$ as follows:
\begin{equation}\small
\begin{aligned}
&\max_{P^r} \Phi=\prod_{i=1}^{N}\left(\frac{\log(1+L_1^iP^r)}{L_3^i((1+L_1^iP^r)^{L_2^i}-1)+L_4^i}\right)\left(K_1-K_2P^r\right),\\
&~~s.t.~0\leq P^r\leq \min\left\{P_0^r,\{\frac{\sigma^2}{|g_i|^2}(1+\frac{P_0^s|h_i|^2}{\sigma^2})^{\frac{\beta_i}{\gamma_i}}-\frac{\sigma^2}{|g_i|^2}\}\right\},\\
\end{aligned}
\label{equ:quasipowerproof}
\end{equation}
where for any $1\leq i\leq N$, $L_1^i$ to $L_4^i$ and $K_1$ to $K_2$ are given as follows:
\begin{equation}\small
\left\{
\begin{aligned}
&L_1^i=\frac{|g_i|^2}{\sigma^2}\geq0,\\
&L_2^i=\frac{\gamma_i}{\beta_i}\geq0,\\
&L_3^i=\beta_iT\frac{\sigma^2}{{h_i}^2} \geq 0,\\
&L_4^i=\alpha_iTP_0^s\geq0,\\
&K_1=\eta\sum_{i=1}^N\alpha_iP_0^sT-E_0 \geq 0,\\
&K_2=\sum_{i=1}^N\gamma_iT\geq 0,\\
\end{aligned}
\right.
\label{equ:proofpoweralgebra}
\end{equation}
Since $\Phi$ in (\ref{equ:quasipowerproof}) is a smooth function when $P^r\geq0$, we prove the quasi-concavity of (\ref{equ:quasipowerproof}) through \textbf{Step~1}  to \textbf{Step~2} as follows.
\begin{step}
The problem in (\ref{equ:quasipowerproof}) can be decomposed as follows:
\begin{equation}\small
\begin{aligned}
&\max_i \log(\Phi)\\
&~~=\sum_{i=1}^{N}\log\left(\frac{\log(1+L_1^iP^r_i)}{L_3^i((1+L_1^iP^r_i)^{L_2^i}-1)+L_4^i}\right)+\log\left(K_1-K_2P^r_1\right),\\
&~~s.t.~0\leq P^r\leq \min\left\{P_0^r,\{\frac{\sigma^2}{|g_i|^2}(1+\frac{P_0^s|h_i|^2}{\sigma^2})^{\frac{\beta_i}{\gamma_i}}-\frac{\sigma^2}{|g_i|^2}\}\right\},\\
&~~~~~~~P^r_i=P^r_1.\\
\end{aligned}
\label{equ:quasipowerdecomposition}
\end{equation}
\end{step}
\begin{step}
At this step, we prove the decomposed problem in (\ref{equ:quasipowerdecomposition}) to be a quasi-concave problem. The quasi-concavity of the norm $\log\left(\frac{\log(1+L_1^iP^r_i)}{L_3^i((1+L_1^iP^r_i)^{L_2^i}-1)+L_4^i}\right)$ for variable $P^r_i$ is easy to test when $i\neq1$. When $i=1$, we prove the quasi-concavity of the objective $\Phi$ for variable $P^r_1$ through proving Lemma~\ref{lem:quasipowerproof1} and Lemma~\ref{lem:quasipowerproof2} as follows:
\begin{lemma}
$\nabla_{P^r_1}\Phi\geq0$ when $P^r_1=0$ and  $\nabla_{P^r_1}\Phi\leq0$ when $P^r_1=\frac{K_1}{K_2}$, where $\nabla_{P^r_1}\Phi$ is the first-order derivative of $\Phi$ for variable $P^r_1$.
\label{lem:quasipowerproof1}
\end{lemma}
\begin{proof}
$\nabla_{P^r_1}\Phi$ is calculated as follows:
\begin{equation}\small
\begin{aligned}
\nabla_{P^r_1}\Phi=\Phi\cdot&\Bigg(\frac{L_1^1}{(1+L_1^1P^r)\log(1+L_1^1P^r_1)}\\
&-\frac{L_3^1L_2^1L_1^1(1+L_1^1P^r_1)^{L_2^1}}{(L_3^1((1+L_1^1P^r_1)^{L_2^1}-1)+L_4^1)(1+L_1^1P^r_1)}\\
&-\Phi\cdot\frac{K_2}{K_1-K_2P^r_1}\Bigg)\\
\end{aligned}
\end{equation}
It is easy to test that $\nabla_{P^r_1}\Phi\geq0$ when $P^r=0$ and  $\nabla_{P^r_1}\Phi\leq0$ when $P^r=\frac{D_1}{D_2}$.
\end{proof}
\end{step}

\begin{lemma}
There does not exist any local minimum inflection point of $\Phi$ when $0< P^r_1\leq\frac{K_1}{K_2}$. That is to say, $\nabla^2_{P^r_1}\Phi\geq0$ guarantees $\nabla_{P^r_1}\Phi\leq 0$, where $\nabla^2_{P^r_2}\Phi\geq0$ is the second-order derivative function of $\Phi$.
\label{lem:quasipowerproof2}
\end{lemma}
\begin{proof}
We rewrite $\nabla_{P^r_1}\Phi$ as follows:
\begin{equation}\small
\left\{
\begin{aligned}
\nabla_{P^r_1}\Phi&=\Phi\cdot\left(\zeta_1-\frac{K_2}{K_1-K_2P^r_1}\right),\\
\zeta_1&=\frac{L_1^1}{(1+L_1^1P^r)\log(1+L_1^1P^r_1)}\\
&-\frac{L_3^1L_2^1L_1^1(1+L_1^1P^r_1)^{L_2^1}}{(L_3^1((1+L_1^1P^r_1)^{L_2^1}-1)+L_4^1)(1+L_1^1P^r_1)}.\\
\end{aligned}
\right.
\end{equation}
Thus, we calculate the $\nabla^2_{P^r_1}\Phi$ as follows:
\begin{equation}\small
\begin{split}
\nabla^2_{P_1^r}\Phi=&\Phi\cdot\left(\zeta_1-\frac{K_2}{K_1-K_2P^r_1}\right)^2\\
&-\Phi\cdot\nabla_{P^r_1}\left(\zeta_1-\frac{K_2}{K_1-K_2P^r_1}\right)\\
=&-\Phi\cdot\frac{L_1^1}{1+L_1^1P^r_1}\zeta_1\\
&-2\Phi\cdot\frac{L_3^1L_2^1L_1^1(1+L_1^1P^r_1)^{L_2^1}}{(L_3^1((1+L_1^1P^r_1)^{L_2^1}-1)+L_4^1)(1+L_1^1P^r_1)}\zeta_1\\
&-2\Phi\cdot\frac{K_2}{K_1-K_2P^r_1}\zeta_1\\
&-\Phi\cdot\frac{L_3^1(L_2^1)^2(L_1^1)^2(1+L_1^1P^r_1)^{L_2^1}}{(L_3^1((1+L_1^1P^r_1)^{L_2^1}-1)+L_4^1)(1+L_1^1P^r_1)^2}\\
\leq&-\Phi\cdot\frac{L_1^1}{1+L_1^1P^r_1}\zeta_1\\
&-2\Phi\cdot\frac{L_3^1L_2^1L_1^1(1+L_1^1P^r_1)^{L_2^1}}{(L_3^1((1+L_1^1P^r_1)^{L_2^1}-1)+L_4^1)(1+L_1^1P^r_1)}\zeta_1\\
&-2\Phi\cdot\frac{K_2}{K_1-K_2P^r_1}\zeta_1\\
=&-\Phi\cdot\xi_1\cdot\zeta_1,\\
\end{split}
\label{equ:lem41}
\end{equation}
where $\xi_1>0$ is given as follows:
\begin{equation}\small
\begin{aligned}
\xi_1&=\frac{L_1^1}{1+L_1^1P^r_1}\zeta_1+2\cdot\frac{L_3^1L_2^1L_1^1(1+L_1^1P^r_1)^{L_2^1}}{(L_3^1((1+L_1^1P^r_1)^{L_2^1}-1)+L_4^1)(1+L_1^1P^r_1)}\\
&+2\cdot\frac{K_2}{K_1-K_2P^r_1}.\\
\label{equ:lem42}
\end{aligned}
\end{equation}
When $\nabla^2_{P^r_2}\Phi\geq0$, we have $-\Phi\cdot\xi_1\cdot\zeta_1\geq0$. Since $\Phi>0$ and $\xi_1>0$, we have $\zeta_1\leq0$. Therefore, we have the inequation as follows:
\begin{equation}
\begin{aligned}
\nabla_{P^r_1}\Phi&=\Phi\cdot(\zeta_1-\frac{K_2}{K_1-K_2P^r_1})\leq \Phi\cdot\zeta_1 \leq0.\\
\end{aligned}
\end{equation}
Hence, we conclude that $\nabla^2_{P^r_2}\Phi\geq0$ guarantees $\nabla_{P^r_1}\Phi\leq 0$.
\end{proof}
According to the decomposition in \textbf{Step~1}, Lemma~\ref{lem:quasipowerproof1} and Lemma~\ref{lem:quasipowerproof2} in \textbf{Step~2} and the linear constraint in~(\ref{equ:quasipowerproof}), we conclude that the problem in~(\ref{equ:obj}) is a quasi-concave problem with variables $\{P_i^{s0}\}$ and $P^r$.

\section{Proof of Theorem \ref{the:concavetime}}
\label{app:concavetime}
According to Lemma~\ref{lem:datatrans} and Corollary~\ref{cor:ehtime}, we can rewrite the problem~in~(\ref{equ:obj}) with variables $\alpha$ and $\{\gamma_i\}$ as follows:
\begin{equation}\small
\begin{aligned}
&\max_{\{\alpha_i\},\{\gamma_i\}} \Phi=\left(\prod_{P_i^{s1}=P_0^s} \frac{D_1^i\gamma_i}{D_2^i\alpha+D_3^i\gamma_i}\right)\cdot(F_1\alpha-F_2-\sum_{i=1}^ND_4^i\gamma_i),\\
&~~s.t.~\alpha+\sum_{i=1}^N(\gamma_i+D_5^i\gamma_i)=1,\\
&~~~~~~~(\gamma_i+D_5^i\gamma_i)T\geq\theta_0T,\\
\end{aligned}
\label{theproof:timequasi}
\end{equation}
where $\{D_1^i\}$ to $\{D_5^i\}$ and $F_1$ to $F_2$ are positive, which are given as follows:
\begin{equation}
\left\{
\begin{aligned}
&D_1^i=T\log\left(1+P^r\frac{|g_i|^2}{\sigma^2}\right),\\
&D_2^i=P_i^{s1}T,\\
&D_3^i=P_i^{s0}T\frac{\log(1+P^r|g_i|^2/\sigma^2)}{\log(1+P_i^{s0}|h_i|^2/\sigma^2)},\\
&D_4^i=P^rT,\\
&D_5^i=\frac{\log(1+P^r|g_i|^2/\sigma^2)}{\log(1+P_i^{s0}|h_i|^2/\sigma^2)},\\
&F_1=\eta\sum_{P_i^{s1}=P_0^s}P_0^{s}T|h_i|^2,\\
&F_2=E_0,\\
\end{aligned}
\right.
\label{equ:prooftimealgebra}
\end{equation}
where $D_1^i$ to $D_5^i$ and $F_1$ to $F_2$ are independent from $\alpha$ and $\{\gamma_i\}$ and $\beta_i=D_5^i\gamma_i$.
Since $\gamma_i$ is separable in~(\ref{theproof:timequasi}), it is easy to test that $\Phi$ is a concave function for each $\gamma_i$. Then, the constraint in~(\ref{theproof:timequasi}) are linear constraints. Thus, the problem~in~(\ref{theproof:timequasi}) is a quasi-concave problem when $\Phi$ is a quasi-concave function of $\alpha$, which we prove it in Lemma~\ref{lem:quasialpha} as follows:
\begin{lemma}
$\nabla_\alpha^2\Phi\geq0$ guarantees $\nabla^1_\alpha\Phi<0$, where $\nabla_\alpha^i$ is the $i^{th}$-order derivative with variable $\alpha$. That is to say, $\Phi$ has no local minimum inflection point of $\alpha$, which indicates that the problem~(\ref{theproof:timequasi}) is a quasi-concave problem of $\alpha$.
\label{lem:quasialpha}
\end{lemma}
\begin{proof}
We can calculate $\nabla_\alpha^2\Phi$ and $\nabla^1_\alpha\Phi$, respectively, as follows:
\begin{equation}\small
\left\{
\begin{aligned}
&\nabla^1_\alpha\Phi=\left(-\sum_{P_i^{s1}=P_0^s}\mu_i+\nu\right)\cdot\Phi,\\
&\nabla^2_\alpha\Phi=\left(-\sum_{P_i^{s1}=P_0^s}\mu_i+\nu\right)^2\cdot\Phi+\left(\sum_{P_i^{s1}=P_0^s}(\mu_i)^2-\nu^2\right)\cdot\Phi,\\
&\mu_i=\frac{G_1^i}{G_1^i\alpha+G_2^i}>0,~~\nu=\frac{H_1}{H_1\alpha-H_2}>0,\\
\end{aligned}
\right.
\end{equation}
where $\{G_1^i\}$, $\{G_2^i\}$, $H_1$ and $H_2$ are independent from $\alpha$, which are given as follows:
\begin{equation}
\left\{
\begin{aligned}
&G_1^i=D_2^i=D_2^i=P_i^{s1}T,\\
&G_2^i=D_3\gamma_i=P_i^{s0}\gamma_iT\frac{\log(1+P_i^{s0}|h_i|^2/\sigma^2)}{\log(1+P^r|g_i|^2/\sigma^2)},\\
&H_1=F_1=\eta\sum_{P_i^{s1}=P_0^s}P_0^{s}T|h_i|^2,\\
&H_2=F_2+\sum_{k=1}^ND_4^k\gamma_k=E_0+\sum_{k=1}^ND_4^k\gamma_k.\\
\end{aligned}
\right.
\end{equation}
Therefore, when $\nabla^2_\alpha\Phi\geq0$, we have the inequation as follows:
\begin{equation}\small
\begin{aligned}
&0\leq \nabla^2_\alpha\Phi\\
&=\left(\left(\sum_{P_i^{s1}=P_0^s}\mu_i\right)^2+\sum_{P_i^{s1}=P_0^s}(\mu_i)^2-2\left(\sum_{P_i^{s1}=P_0^s}\mu_i\right)\nu\right)\Phi\\
&<\left(2\left(\sum_{P_i^{s1}=P_0^s}\mu_i\right)^2-2\left(\sum_{P_i^{s1}=P_0^s}\mu_i\right)\nu\right)\Phi\\
&=-2\left(\sum_{P_i^{s1}=P_0^s}\mu_i\right)\cdot\nabla^1_\alpha \Phi\Rightarrow\nabla^1_\alpha \Phi<0.\\
\end{aligned}
\end{equation}
Thus, the problem~in~(\ref{theproof:timequasi}) is a quasi-concave problem of $\alpha$. Therefore, the problem in (\ref{equ:obj}) is a quasi-concave problem with variables $\alpha$, $\{\beta_i\}$, and $\{\gamma_i\}$.\par
\end{proof}

\section{Gradient Ascent for Data Transmission Power Optimization}
\label{app:poweropt}
The Lagrangian function for data transmission power optimization problem is given as follows:
\begin{equation}\small
\begin{aligned}
&\mathcal{L}_1^k(P^r,\rho_1,\rho_2)\\
&=-\prod_{i=1}^{N}\left(\frac{\log(1+L_1^iP^r)}{L_3^i((1+L_1^iP^r)^{L_2^i}-1)+L_4^i}\right)\left(K_1-K_2P^r\right)\\
&+\rho_1(0-P^r)+\rho_2(P^r-[P_0^r]^*).\\
\end{aligned}
\end{equation}
Thus ,the dual function is calculated as follows~\cite{convex-2004}\cite{convex-2011}:
\begin{equation}\small
\begin{aligned}
&g_1^k(\rho_1,\rho_2)=\inf_{P^r}\mathcal{L}_1^k(P^r,\rho_1,\rho_2)\\
&=\prod_{i=1}^{N}\left(\frac{\log(1+L_1^i[P^r]^\sharp)}{L_3^i((1+L_1^i[P^r]^\sharp)^{L_2^i}-1)+L_4^i}\right)\left(K_1-K_2[P^r]^\sharp\right)\\
&~~-\rho_2[P_0^r]^*,\\
\end{aligned}
\end{equation}
where $[P^r]^{\sharp}$ equals to $-\rho_2$.\footnote{Since all variables and all parts in the Lagrangian function are real, we eliminate all the conjugate operations.}\par
The dual problem is as follows:
\begin{equation}
\max_{\rho_1,\rho_2}g_1^k(\rho_1,\rho_2),
\end{equation}
with variables $\rho_1$ and $\rho_2$. We can achieve a dual optimal point $\{\rho_1^*,\rho_2^*\}$ through solve the dual problem above.
Since the data transmission power optimization problem in (\ref{equ:quasipowermain}) is a quasi-concave optimization problem, instead of a strict concave optimization problem, we can only 'possibly' obtain
a primal optimal point $(P^r)^*$ from a dual optimal point $\{\rho_1^*,\rho_2^*\}$ as
\begin{equation}
(P^r)^*=\mbox{arg}\min_{P^r} L_1^k(P^r,\rho_1^*,\rho_2^*).\\
\label{equ:lagpowerdetial}
\end{equation}
Since the $g_1^k$ is differentiable when $\rho_1>0$ and $\rho_2>0$, we can update the primal variable $P^r$ and the Lagrange multipliers $\rho_1$ and $\rho_2$ as follows~\cite{convex-2004}:
\begin{equation}
\left\{
\begin{aligned}
&(P^r)^{t+1}=\mbox{arg}\min_{P^r}L_1^k(P^r,\rho_1^t,\rho_2^t),\\
&\rho_1^{t+1}=\rho_1^{t}+\varsigma_1^{t}\left(0-(P^r)^{t+1}\right),\\
&\rho_2^{t+1}=\rho_2^{t}+\varsigma_2^{t}\left((P^r)^{t+1}-[P_0^r]^*\right),\\
\end{aligned}
\right.
\end{equation}
where $\varsigma_1^{t}>0$ and $\varsigma_2^{t}>0$ are the step variables for the $t^{th}$ step. \par
The iteration ends (the gradient ascent algorithm converges) at $t_1^{th}$ iteration when the primal variable $P^r$ satisfies the linear stopping criterion~\cite{convex-2004}, which is given as follows:
\begin{equation}\small
|(P^r)^{t_1+1}-(P^r)^{t_1}|<\varepsilon_1,
\label{equ:poweraccuracy}
\end{equation}
where $\varepsilon_1$ is a minor parameter to determine the convergence. \par Finally, we obtain the optimal $(P^r)^*$ in~(\ref{equ:lagpowerdetial}) as $(P^r)^*=(P^r)^{t_1+1}$. We emphasize that since the data transmission power optimization problem is not strictly concave, we can obtain a suboptimal solution with dual ascent method.\par
\vspace{-2mm}
\section{Gradient Ascent for Time Division}
\label{app:timeopt}
Similarly with the gradient ascent method in Appendix~\ref{app:poweropt}, we can update the time division variables $\alpha$, $\{\gamma_i\}$ and the corresponding Lagrange multipliers $\varrho_1$ and $\{\varrho_2^i\}$ in~(\ref{equ:lagtime}) as follows:
\begin{equation}
\left\{
\begin{aligned}
&(\alpha,\{\gamma_i\})^{t+1}=\mbox{arg}\min_{\alpha,\{\gamma_i\}}L_2^k(\alpha,\{\gamma_i\},\varrho_1^t,(\{\varrho_2^i\})^{t}),\\
&\varrho_1^{t+1}=\varrho_1^{t}+\tau_1^{t}\left(\alpha^{t+1}+\sum_{i=1}^N(\gamma_i^{t+1}+D_5^i\gamma_i^{t+1})-1\right),\\
&(\varrho_2^i)^{t+1}=(\varrho_2^i)^{t}+(\tau_2^i)^{t}\cdot\left(\theta_0T-(\gamma_i^{t+1}+D_5^i\gamma_i^{t+1})T\right),\\
\end{aligned}
\right.
\end{equation}
where $\tau_1^{t}>0$ and $\{(\tau_2^i)^{t}>0\}$ are the step variables for the $t^{th}$ step. Similar with the gradient ascent method in Appendix~\ref{app:poweropt}, we can set a primal stopping criterion when the time division converges at $t_2^{th}$ iteration, which is given as follows:
\begin{equation}
|\alpha^{t_2+1}-\alpha^{t_2}|<\varepsilon_2,
\label{equ:timeaccuracy}
\end{equation}
where $\varepsilon_2$ is a minor parameter to determine the convergence. \par
To avoid misleading the time division process, we obtain the $(\alpha,\{\gamma_i\})^{t+1}$ at each $t^{th}$ iteration, through alternatively update of $\alpha$ and $\{\gamma_i\}$, which is given as follows:
\begin{equation}
\left\{
\begin{aligned}
&\alpha^{\vartheta+1}=\alpha^{\vartheta}-(\kappa_1)^{\vartheta}\nabla_{\alpha} L_2^k(\alpha,(\{\gamma_i\})^{\vartheta},\varrho_1^t,(\{\varrho_2^i\})^{t}),\\
&\gamma_i^{\vartheta+1}=\gamma_i^{\vartheta}-(\kappa_2^i)^{\vartheta}\nabla_{\gamma_i} L_2^k(\alpha^{\vartheta},(\{\gamma_j\})_{j\neq i}^{\vartheta},\gamma_i, \varrho_1^t,(\{\varrho_2^i\})^{t}),\\
\end{aligned}
\right.
\end{equation}
where $(\kappa_1)^{\vartheta}>0$ and $\{(\kappa_2^i)^{\vartheta}\}>0$ are the step size. To guarantee the constraints in (\ref{equ:quasitimemain}) are always satisfied, the following equation must be satisfied at $\vartheta^{th}$ update:
\begin{equation}\small
\left\{
\begin{aligned}
&(\kappa_1)^{\vartheta}\cdot\nabla_{\alpha} L_2^k(\alpha,(\{\gamma_i\})^{\vartheta},\varrho_1^t,(\{\varrho_2^i\})^{t})\\
&+\sum_{i=1}^N (\kappa_2^i)^{\vartheta}\cdot\nabla_{\gamma_i} L_2^k(\alpha^{\vartheta},(\{\gamma_j\})_{j\neq i}^{\vartheta},\gamma_i, \varrho_1^t,(\{\varrho_2^i\})^{t})=0,\\
&\gamma_i^{\vartheta}-(\kappa_2^i)^{\vartheta}\nabla_{\gamma_i} L_2^k(\alpha^{\vartheta},(\{\gamma_j\})_{j\neq i}^{\vartheta},\gamma_i, \varrho_1^t,(\{\varrho_2^i\})^{t})>\frac{\theta_0}{1+D_5^i}.\\
\end{aligned}
\right.
\end{equation}
\par
When the update ends at $\vartheta_1^{th}$ update, we have
\begin{equation}
\left\{
\begin{aligned}
&\alpha^{t+1}=\alpha^{\vartheta_1+1},\\
&\gamma_i^{t+1}=\gamma_i^{\vartheta_1+1},~\forall 1\leq i\leq N.\\
\end{aligned}
\right.
\end{equation}
\par
When the gradient ascent method converges at $t_2^{th}$ iteration, we can obtain the optimal values as $\alpha^{\star}$ and $(\{\gamma_i\})^{\star}$ for the time division problem in~(\ref{equ:quasitimemain}), which is as follows:
\begin{equation}
\left\{
\begin{aligned}
&\alpha^{\star}=\alpha^{t_2+1},\\
&(\{\gamma_i\})^{\star}=(\{\gamma_i\})^{t_2+1}.\\
\end{aligned}
\right.
\end{equation}
\par

\end{document}